\documentclass[11pt,a4paper]{article}

\usepackage[a4paper, total={6in, 8in}]{geometry}

\usepackage{cite}
\usepackage{amsmath,amssymb,amsfonts,amsthm}
\usepackage{comment}
\usepackage{textcomp}
\def\BibTeX{{\rm B\kern-.05em{\sc i\kern-.025em b}\kern-.08em
    T\kern-.1667em\lower.7ex\hbox{E}\kern-.125emX}}

\usepackage[hidelinks]{hyperref}
\usepackage{titlesec}
\usepackage{tikz}
\usepackage{graphicx}
\usepackage{xspace}
\usepackage{subfig}
\usepackage{bbding}
\usepackage{pifont}
\usepackage{booktabs,color,doi,graphicx,latexsym,url,xcolor,xspace}
\usepackage{wrapfig}
\usepackage[linesnumbered,ruled,vlined]{algorithm2e} 
\usepackage{fancyhdr}
\usepackage{enumerate}
\usepackage{balance}

\usepackage{pifont}
\newcommand{\xmark}{\ding{55}}%

\usepackage{colortbl}
\definecolor{lgray}{rgb}{0.95,0.95,0.95}

\newcommand{\opt}{\textsc{Opt}\xspace}
\newcommand{\rank}{\mathrm{rank}\xspace}
\newcommand{\frank}{\mathrm{frnk}\xspace}
\newcommand{\level}{\ell\xspace}
\newcommand{\levopt}{\ell^\textsc{opt}\xspace}

\newcommand{\OPT}{\textsc{Opt}\xspace}
\newcommand{\RAND}{\textsc{Rand}\xspace}
\newcommand{\RandLong}{\textsc{Random-Push}\xspace}
\newcommand{\DET}{\textsc{Rtr}\xspace}
\newcommand{\DetLong}{\textsc{Rotor-Push}\xspace}
\newcommand{\movehalf}{\textsc{Move-Half}\xspace}
\newcommand{\maxpush}{\textsc{Max-Push}\xspace}

\newcommand{\PD}[2]{\mathrm{PD}({#1},{#2})}
\newcommand{\flip}{\mathrm{flip}}
\newcommand{\el}{d}

\newcommand{\node}{\mathrm{nd}\xspace}
\newcommand{\elem}{\mathrm{el}\xspace}
\newcommand{\E}{\mathbf{E}}
\newcommand{\statopt}{\textsc{Static-Opt}\xspace}
\newcommand{\oblstat}{\textsc{Static-Oblivious}\xspace}

\newcommand{\cred}{c}
\newcommand{\credlev}{w^\text{LEV}}
\newcommand{\credpos}{w^\text{FRNK}}
\newcommand{\re}{e^*}

\newcommand{\iosif}[1]{#1}

\newcommand{\alenex}[1]{\textcolor{black}{#1}}

\newtheorem{theorem}{Theorem}
\newtheorem{lemma}[theorem]{Lemma}

\newtheorem{definition}{Definition}
\newtheorem{observation}{Observation}

\begin{document}

\title{Deterministic Self-Adjusting Tree Networks\\ Using Rotor Walks\footnote{This project received funding by the European Research Council (ERC), grant agreement 864228, 
Horizon 2020, 2020-2025, and by Polish National Science Centre grant 2016/22/E/ST6/00499.}}

\author{Chen Avin$^1$~~ 
Marcin Bienkowski$^2$~~
Iosif Salem$^3$ \\
Robert Sama$^4$~~
Stefan Schmid$^5$~~
Paweł Schmidt$^6$\\~\\

{\small $^1$School of Electrical and Computer Engineering,
Ben Gurion University of the Negev, Israel}\\
{\small \texttt{avin@cse.bgu.ac.il}}\\

{\small $^2$Institute of Computer Science, University of Wroclaw, Poland}\\
{\small \texttt{marcin.bienkowski@cs.uni.wroc.pl}}\\

{\small $^3$Department of Telecommunication Systems, TU Berlin, Germany}\\
{\small \texttt{iosif.salem@inet.tu-berlin.de}}\\

{\small $^4$Faculty of Computer Science, University of Vienna, Austria}\\
{\small \texttt{robert.sama@outlook.com}}\\

{\small $^5$Department of Telecommunication Systems, TU Berlin, Germany \& University of Vienna, Austria}\\
{\small \texttt{stefan.schmid@tu-berlin.de}}\\

{\small $^6$Institute of Computer Science, University of Wroclaw, Poland}\\
{\small \texttt{pawel.schmidt@cs.uni.wroc.pl}}
}

\date{}






\maketitle



\begin{abstract}
We revisit the design of self-adjusting single-source tree networks. The problem
can be seen as a~generalization of the classic list update problem to trees, and
finds applications in reconfigurable datacenter networks. 
We are given a balanced binary tree $T$ connecting $n$ nodes
$V=\{v_1,\ldots,v_n\}$. A~source node $v_0$, attached to the root of the tree,
issues communication requests to nodes in~$V$, in an online and adversarial
manner; the access cost of a request to a node $v$, is given by the current
depth of $v$ in $T$. The online algorithm can try to reduce the
access cost by performing swap operations, with which the position of
a node is exchanged with the position of its parent in the tree; a swap
operation costs one unit. The objective is to design an online algorithm which minimizes the total access cost plus adjustment cost (swapping). Avin et al.~\cite{avin2021dynamically}
(LATIN 2020) recently presented \RandLong, a constant competitive
online algorithm for this problem, based on random walks, together with a
sophisticated analysis exploiting the working set property.

This paper studies analytically and empirically, online algorithms for this problem. In particular, we explore 
how to derandomize \textsc{Random-Push}. In the analytical part, we consider a simple derandomized algorithm which we call \textsc{Rotor-Push},
as its behavior is reminiscent of rotor walks. 
Our first contribution is a~proof that \textsc{Rotor-Push} is constant
competitive: its competitive ratio is 12 and hence by a factor of five lower than the best existing competitive ratio. Interestingly, in contrast to \textsc{Random-Push}, the algorithm
does not feature the working set property, which requires a new analysis. 
We further present a significantly improved and simpler analysis for the randomized algorithm, showing that it is 16-competitive. 

In the empirical part, we compare all self-adjusting
single-source tree networks, using both synthetic and real data. In
particular, we shed light on the extent to which these self-adjusting trees can
exploit temporal and spatial structure in the workload. Our
experimental artefacts and source codes are publicly available.
\end{abstract}


\section{Introduction}

One of the initially studied and fundamental online problems is known as the
\emph{list update problem}: There is a set of elements $E=\{e_1,\ldots,e_n\}$
organized in a linked list where the cost of accessing an element is equal to
its distance from the front of the list. Given a~request sequence of accesses
$\sigma=\sigma^{1}, \sigma^{2}, \ldots$, where $\sigma^{t} = e_i \in E$
denotes that element $e_i$ is requested, the problem is to come up with a
strategy of reordering the list so that the total cost of accesses and reordering is minimized.
The basic reordering operation involves swapping two adjacent elements which
costs one unit. 

The problem is inherently online, that is, decisions of an algorithm have to be
made immediately upon arrival of the request and without the knowledge of future
ones. The efficiency of an algorithm is then analyzed by comparing its cost to
the cost of an optimal offline strategy \OPT, and the ratio of these costs,
called \emph{competitive ratio}, is subject to minimization. Many constant
competitive algorithms are known for the list update problem today, most
prominently the \textsc{Move-To-Front} algorithm~\cite{sleator1985amortized} and
its variants~\cite{Albers03,KamaliSurvey2013,Lopez-OrtizRR20,KamaliSurvey2013,albers1997revisiting,Albers2020,Reingold1994,Munro2000}; all basically moving an~accessed element to (or towards) the
front of the list. The prevalent model in the
literature assumes that the movement of an accessed element towards the list
head is free~\cite{sleator1985amortized}, which however affects the achievable competitive ratios only by
constant factors. 

\noindent \textbf{Tree structure.}
This paper revisits the list update problem but replaces the list with 
a~complete and balanced binary tree. That is, there is an underlying and fixed
structure of $n$~\emph{nodes} forming a complete binary tree $T$ and $n$ elements
$E=\{e_1,\ldots,e_n\}$, where each node has to be occupied exactly by one
element. We denote the node currently holding element~$e$ by $\node(e)$ and the unique element 
stored currently at node $v$ by $\elem(v)$. For any node $v$ we denote its 
\emph{tree level} by $\level(v)$, where the root node has level $0$.

Analogously to the list update problem, the access cost to an element $e$ stored
currently at $v = \node(e)$ is given by $\level(v) + 1$, and at a unit cost it is
possible to swap elements $e_i$ and~$e_j$ occupying adjacent nodes (i.e., 
$\node(e_i)$ is the parent of $\node(e_j)$). Again, the objective is to design
an~online algorithm which minimizes the total cost defined as the cost of all
accesses and swaps.

\noindent \textbf{Reconfigurable optical networks.}
Besides being theoretically interesting as a natural generalization of the list
update problem, such self-adjusting single-source tree structures have recently
gained interest due to their applications in reconfigurable optical networks
\cite{avin2021dynamically}. There, the sequence $\sigma=\sigma^{1},
\sigma^{2}, \ldots$ corresponds to communication requests arriving from a~source node 
which is attached to the root node of the tree. 
These single-source tree networks can be combined to form self-adjusting networks 
which serve multiple sources and whose topology can be an arbitrary
degree-bounded graph~\cite{apocs21renets,disc17}. 
Therefore, the insights gained from analyzing single-source tree networks can assist the design of more efficient self-adjusting networks.

\subsection{Previous results}

A natural idea to design self-adjusting balanced tree networks could be to consider an~immediate generalization of the \textsc{Move-To-Front} strategy: upon a
request to element~$e$, we perform swaps along the path from $\node(e)$ to the root node. 
This moves accessed element $e$ to the root node and pushes all remaining elements on this 
path one level down. However, it is easy to observe~\cite{avin2021dynamically} that this
solution would yield a competitive ratio of $\Omega(\log n / \log \log
n)$\footnote{Note that the ratio of $O(\log n)$ is trivially achievable by an
algorithm that performs no swaps as each access incurs cost at least $1$ to
$\OPT$ and at most $O(\log n)$ (tree depth) to an online algorithm.}: If
$\sigma$ only consists of the elements along the path which are accessed in a
round robin manner, always requesting the leaf entails a cost of $\Theta(\log n)$
to such an online algorithm. In contrast, a feasible strategy for \OPT is to place
all these $\Theta(\log n)$ elements in the first $\Theta(\log\log n)$ levels,
resulting in an access cost of $O(\log \log n)$ per request.

To overcome the problem above, Avin et al.~\cite{avin2021dynamically} proposed a
randomized algorithm \textsc{Random-Push}. In a nutshell, it moves the accessed
element $e$ to the root node, but to make space for it, it chooses a
\emph{random path of nodes} starting at the root node and pushes elements on
this path one level down. More precisely, let $v = \node(e)$ and $d =
\level(v)$. $\textsc{Random-Push}$ chooses a \emph{random} node $v'$ uniformly
on level $d$, which induces a random path of nodes $s_1, s_2, \ldots, s_{d-1},
s_d = v'$, where $s_1$ is the root note. Now for a cycle of nodes $s_1 \to s_2
\to \ldots \to s_{d-1} \to v' \to v \to s_1$, each of the corresponding elements
is moved to the next node on the cycle. (That is, for $i < d$, an 
element $e = \elem(s_i)$ is \emph{pushed down} by one level to a \emph{random}
child of $\node(e)$.) It is easy to observe (cf.~Section~\ref{sec:preliminaries})
that the cyclic-shift of elements can be executed using $O(\level(v))$ swaps of
adjacent elements.


By a careful analysis of working set properties of the algorithm, Avin et
al.~\cite{avin2021dynamically} showed that \textsc{Random-Push} is
$O(1)$-competitive. Specifically, their analysis revolved around the notion of a
Most Recently Used (MRU) tree (where for any two nodes $u$ and $v$, if $u$ 
was accessed more recently than $v$, then it is not further away from the root than $v$).
Such a tree has \emph{the working set property}: the cost of
accessing element $e$ at time $t$ depends logarithmically on the number of
distinct items accessed since the last access of $e$ prior to time $t$,
including $e$. Avin et al.~\cite{avin2021dynamically} showed that the working
set bound is a cost lower bound for any (also offline) algorithm, and proved that
\textsc{Random-Push} approximates (in expectation) 
an MRU tree at any time requiring low swapping costs.


\subsection{Our contribution}

This paper studies whether the \textsc{Random-Push} algorithm can be
derandomized while maintaining the constant competitive ratio. We propose a
natural approach to imitate the random walk executed implicitly by
\textsc{Random-Push} by the following \emph{rotor
walk}~\cite{priezzhev1996eulerian,holroyd2010rotor,cooper2004simulating,akbari2013parallel,cooper-soda-08}.
In our approach, each non-leaf node in the binary tree maintains a two-state
pointer pointing to one of its two children. Whenever an element stored at this
node is pushed down, the direction is according to this pointer and, right after
that, the pointer is toggled, now pointing at the other child node. 

Perhaps surprisingly, it turns out that this algorithm, to which we refer to as
\textsc{Rotor-Push}, has fairly different properties from the algorithm based on
random walks. In particular, unlike \textsc{Random-Push}, an
adversary can fool \textsc{Rotor-Push} so that it does not fulfill the working
set property: using \textsc{Rotor-Push}, the depth of a node can be as high as
linear in its working set size (see Lemma~\ref{lem:rotorWS}, Section \ref{sec:lack}), 
while for \textsc{Random-Push} it was at most logarithmic.

\begin{table*}
\centering
\begin{tabular}{ccccc}
 \textbf{Algorithm} & Access Cost: & Total Cost: & Deterministic  & Competitive Ratio \\
 &  WS Property & WS Bound &   & \\
\hline\hline 
\rowcolor{lgray}
 Random-Push~\cite{avin2021dynamically}  &  \checkmark &  \checkmark &  \xmark &    60, {\color{blue} \textbf{16} (Thm.~\ref{thm:16})}\\
Move-Half~\cite{avin2021dynamically}  & \xmark & \checkmark & \checkmark & 64 \\
\rowcolor{lgray}
Strict-MRU~\cite{avin2021dynamically} & \checkmark & ? & \checkmark & ? \\
{\color{blue} \textbf{Rotor-Push}}  &  {\color{blue} \textbf{\xmark} (Lem.~\ref{lem:rotorWS})} &  {\color{blue} \textbf{?}}  &  {\color{blue} \textbf{\checkmark}} &  {\color{blue} \textbf{12} (Thm.~\ref{thm:12})}\\
\hline
\end{tabular}
\caption{Overview of algorithms and the properties each has (\checkmark) or not (\xmark); question marks refer to open problems. In blue the new results of this paper.}
\label{tbl:summary}
\end{table*}

Despite these differences, we show that the deterministic \textsc{Rotor-Push}
algorithm still achieves a constant competitive ratio. Specifically, we show
that \textsc{Rotor-Push} achieves a~competitive ratio of 12, while the best
known existing competitive ratio was 60 (achieved by \textsc{Random-Push}): 
a~factor of 5 improvement. 
Compared to \textsc{Move-Half}, the currently best deterministic algorithm also presented in~\cite{avin2021dynamically}, the improvement is even larger. 
To derive this result, we present a novel analysis. We
show how to reuse our techniques to provide a~significantly simpler analysis of
the constant-competitive ratio of \textsc{Random-Push}, also improving the
competitive ratio from 60 to 16.

Our second contribution is an empirical study and comparison of self-adjusting
single-source tree networks, using both synthetic and real data. In
particular, we shed light on the extent to which these self-adjusting trees can
exploit temporal and spatial structure in the workload. 
Our experimental artefacts and source codes are publicly available \cite{repo}.

Table~\ref{tbl:summary} summarizes the properties of the different algorithms studied in this paper (details will follow).
In bold blue we highlight our contributions in this paper. 


\subsection{Related work}

Our work considers a generalization of the list access problem to trees.
Previous work on self-adjusting trees primarily focused on binary search trees
(BSTs) such as splay trees \cite{sleator1985self}. In contrast to our model,
self-adjustments in BSTs are based on \emph{rotations} (which are assumed to
have constant cost). While self-adjusting binary search trees such as splay trees
have the working set property,
it is still unknown whether they are constant
competitive. Our model differs from this line of research in that our trees are
not searchable and the working set property implies constant competitiveness, as shown
in~\cite{avin2021dynamically}. 
Non-searchable trees have already been studied in
a model where trees can be changed using rotations, and it is known that
existing lower bounds for (offline) algorithms on BSTs also apply to
rotation-based unordered trees~\cite{fredman2012generalizing}. This
correspondence between ordered and unordered trees however no longer holds under
weaker measures~\cite{iacono2005key}. In contrast to rotation models, the swap
operations considered in our work do not automatically pull subtrees along,
which renders the problem different. 

In previous work, Avin et al. \cite{avin2021dynamically} presented the first constant-competitive
online algorithms for self-adjusting tree networks. In addition to
\textsc{Random-Push} which provides probabilistic guarantees, they also
presented a constant-competitive deterministic algorithm \textsc{Move-Half} (\iosif{cf. Algorithm \ref{alg:movehalf}}) and
introduced the notion of \textsc{Strict-MRU} which stores nodes in MRU order, i.e.,
keeps more recently accessed elements closer to the root.\footnote{The authors called the corresponding
algorithm \maxpush (cf. Algorithm \ref{alg:maxpush}).} While \textsc{Strict-MRU} provides optimal access
costs, it is currently not known how to maintain MRU order deterministically and 
efficiently, i.e., at low
swapping cost. Our paper is motivated by the observation that a rotor walk
approach to derandomize  \textsc{Random-Push} on the one hand provides a simple
and elegant algorithm, but at the same time does not ensure the working set
property. 

Rotor walks have received much attention over the last years and are known under different names, e.g., Eulerian walker~\cite{priezzhev1996eulerian}, edge ant walk~\cite{wagner1999distributed}, 
whirling tour~\cite{dumitriu2003playing}, Propp machines~\cite{holroyd2010rotor}, rotor routers~\cite{landau2009rotor}, or deterministic random walks~\cite{friedrich2010cover}.
Their appeal stems from the remarkable similarity to the expectation of random walks, and their resulting application domains, including load-balancing~\cite{akbari2013parallel}.


\section{Preliminaries}
\label{sec:preliminaries}

We are given a complete binary tree $T$ of $n$ nodes. 
Slightly abusing notation, we use $T$ also 
to denote the set of all tree nodes. 
There is a set $E$ of $n$ elements and an algorithm has to maintain a bijective 
mapping $\node: E \to T$. An inverse of function $\node$ is denoted $\elem$. 

\noindent \textbf{Nodes and levels.}
We denote the tree root by $r_T$. For a node $u$, we denote the subtree rooted
at $u$ by $T[u]$. The levels of $T$ are numbered from~$0$, i.e., the only node
at level $0$ is $r_T$. We denote the maximal level in $T$ by $L_T$. 
We extend the notion of levels to elements, $\level(e) = \level(\node(e))$; note
that the level of a node is fixed, while the level of an element may change as
the algorithm rearranges elements in $T$.

\noindent \textbf{Costs.}
There are two types of costs incurred by any algorithm, when serving a single request:
\begin{itemize}
\item Whenever an element $e$ is accessed, an algorithm pays $\level(e) + 1$.
\item Afterwards, an algorithm may perform an arbitrary number of swaps, each 
of cost $1$ and involving two elements occupying adjacent nodes.
\end{itemize}
\noindent \textbf{Arbitrary swaps.}
Assuming that an algorithm can swap two \emph{arbitrary} adjacent elements at
cost only $1$ is rather controversial: this would require a random access to
arbitrary tree nodes. We resolve this issue by making such swaps possible only
for \OPT\footnote{It is worth noting that the existing analysis of
\RandLong~\cite{avin2021dynamically} explicitly forbids \OPT to make such
arbitrary swaps.} (potentially making it unrealistically strong) and using swaps
only in a limited manner in our algorithms. That is, in a single round, 
whenever we access some element (and pay the corresponding access cost), we 
mark all elements on the access path. Subsequent swaps in this round are allowed only
if one of the swapped nodes is marked; after the swap we mark both involved nodes.

\noindent \textbf{Working set bound and working set property.}
Given a sequence $\sigma=\sigma^1,\sigma^2,\ldots$, the \emph{working set} of an
element $e$ at round $t$ is the set of distinct elements (including~$e$) accessed since the last
access of $e$ before round $t$. We call the size of this working set 
the \emph{rank} of $e$, and denote it as $\rank^{(t)}(e)$. 
We drop superscript $(t)$ when it is clear from 
context. The \emph{working set bound} of sequence
$\sigma$ of $m$ requests is defined as $WS(\sigma) = \sum_{t=1}^{m}
\log(\rank^{(t)}(\sigma^t))$. In~\cite{avin2021dynamically}, the authors proved
that, up to a constant factor, 
the working set bound is a lower bound on the cost of any algorithm, even the optimal one.

We say that a self-adjusting tree has the \emph{working set property} if the cost of
each \emph{access} of an element $v$ is \emph{logarithmic} in the element's rank.
The working set property is hence stricter than the working set bound, which considers
the total cost only. 
Any algorithm with the working set property also has the working set bound (if we ignore swapping
cost) and therefore is constant-competitive 
(this is for instance the case for \textsc{Random-Push}~\cite{avin2021dynamically}). 
However, the working set property does not directly imply the working set bound if we account also for the swapping cost: the implication only holds if the reconfiguration cost is proportional to the access cost.

That said, perhaps surprisingly at first sight, online algorithms can also be optimal without the working set property, as we for example demonstrate with \textsc{Rotor-Push}. 

\noindent \textbf{Augmented push-down operation.}
The following operation will be a main building block of the presented algorithms.

\begin{definition}
Fix a tree level $d$ and two $d$-level nodes $u, v$. The
\emph{augmented push-down} operation $\PD{u}{v}$ rearranges the elements as
follows. Let $r_T=v_0, v_1, \dots, v_{\el-1}, v_\el = v$ be the simple path from
root $r_T$ to $v$. Then, we fix a cycle of nodes: $v_0 \to v_1 \to \dots
\to v_{\el-1} \to v_\el \to u \to v_0$ and for each element at a cycle
node, we move it to the next node of the cycle.\footnote{\iosif{Note that $v_d \to u$ and $u \to v_0$ represent the unique paths between those nodes.}}
\end{definition}

In the next section we show that the augmented push-down operation can be 
implemented effectively, using $O(d)$ swaps.


\section{Algorithms}
\label{sec:algorithms}

This section introduces our randomized and deterministic algorithms. 
To this end, we will apply our augmented push-down operation and 
derive first analytical insights. 

\noindent \textbf{Randomized algorithm.}
We start with the definition of a randomized algorithm \RandLong
(\RAND)~\cite{avin2021dynamically}. Upon a request to a $\el^*$-level element
$\re$, \RandLong chooses node~$v$ uniformly at random among all $\el^*$-level
nodes (including $\node(e^*)$) and rearranges the elements by executing the augmented
push-down operation $\PD{\node(\re)}{v}$.

\noindent \textbf{Rotor pointers.}
The random $\el^*$-level node chosen by \RandLong can be picked as a~result of
$\el^*$ independent left-or-right choices. A natural derandomization of this
approach would be to make these choices completely deterministic, i.e., to
maintain a \emph{rotor pointer} at each non-leaf node, pointing to one of its
children (initially to the left one). 
Informally speaking, we will use such a pointer instead of
a random choice and toggle the pointer right after it has been used. 

In a tree $T$, given a current state of pointers, we define a \emph{global
path}, denoted $P^T$, as the root-to-leaf path obtained by starting at $r_T$ and
following the pointers. We denote the unique $\el$-level node of $P^T$ by
$P_\el^T$. To describe our deterministic algorithm, we define a flip operation that
updates the pointers along the global path.

\begin{definition}[Flip]
Fix a tree level $\el$. The operation
$\flip^T(\el)$ toggles pointers at all nodes $P^T_{\el'}$ for $\el' < \el$.
\end{definition}

\begin{figure*}[t]
    \centering
    \begin{minipage}[c]{0.4\linewidth}
        \tikzstyle{rect}=[fill=white, draw=black, shape=rectangle]

\tikzstyle{arr}=[-{latex}]
\tikzstyle{edge}=[-, draw={gray}]

\tikzset{every node/.style={font=\scriptsize}}

\begin{tikzpicture}[scale=0.75]

    \node [style=rect] (0) at (-5.5, -1) {$e_4$};
    \node [style=rect] (1) at (-6, -2) {$e_8$};
    \node [style=rect] (2) at (-5, -2) {$e_9$};
    \node [style=rect] (3) at (-3.5, -1) {$e_5$};
    \node [style=rect] (4) at (-4, -2) {$e_{10}$};
    \node [style=rect] (5) at (-3, -2) {$e_{11}$};
    \node [style=rect] (6) at (-4.5, 0) {$e_2$};
    \node [style=rect] (7) at (-1.5, -1) {$e_6$};
    \node [style=rect] (8) at (-2, -2) {$e_{12}$};
    \node [style=rect] (9) at (-1, -2) {$e_{13}$};
    \node [style=rect] (10) at (0.5, -1) {$e_7$};
    \node [style=rect] (11) at (0, -2) {$e_{14}$};
    \node [style=rect] (12) at (1, -2) {$e_{15}$};
    \node [style=rect] (13) at (-0.5, 0) {$e_3$};
    \node [style=rect] (15) at (-2.5, 1) {$e_1$};

    \node  (18) at (-6, -2.5) {$0$};
    \node  (19) at (-5, -2.5) {$4$};
    \node  (21) at (-2, -2.5) {$1$};
    \node  (22) at (-4, -2.5) {$2$};
    \node  (24) at (0, -2.5)  {$3$};
    \node  (26) at (-1, -2.5) {$5$};
    \node  (27) at (-3, -2.5) {$6$};
    \node  (28) at (1, -2.5)  {$7$};

    \node  (29) at (-5.5, -1.5) {$0$};
    \node  (30) at (-1.5, -1.5) {$1$};
    \node  (31) at (-3.5, -1.5) {$2$};
    \node  (32) at (0.5, -1.5) {$3$};
    
    \node  (33) at (-4.5, -0.5) {$0$};
    \node  (34) at (-0.5, -0.5) {$1$};
    \node  (35) at (-2.5, 0.5) {$0$};

    \draw [style=arr] (0) to (1);
    \draw [style=edge] (0) to (2);
    \draw [style=arr] (3) to (4);
    \draw [style=edge] (3) to (5);
    \draw [style=arr] (6) to (0);
    \draw [style=edge] (6) to (3);
    \draw [style=arr] (7) to (8);
    \draw [style=edge] (7) to (9);
    \draw [style=arr] (10) to (11);
    \draw [style=edge] (10) to (12);
    \draw [style=arr] (13) to (7);
    \draw [style=edge] (13) to (10);
    \draw [style=arr] (15) to (6);
    \draw [style=edge] (15) to (13);
\end{tikzpicture}
    \end{minipage}
    \qquad
    \begin{minipage}[c]{0.4\linewidth}
        \tikzstyle{rect}=[fill=white, draw=black, shape=rectangle]

\tikzstyle{arr}=[-{latex}]
\tikzstyle{edge}=[-, draw={gray}]

\tikzset{every node/.style={font=\scriptsize}}

\begin{tikzpicture}[scale=0.75]
    \node [style=rect] (0) at (-5.5, -1) {$e_2$};
    \node [style=rect] (1) at (-6, -2) {$e_8$};
    \node [style=rect] (2) at (-5, -2) {$e_9$};
    \node [style=rect] (3) at (-3.5, -1) {$e_5$};
    \node [style=rect] (4) at (-4, -2) {$e_{10}$};
    \node [style=rect] (5) at (-3, -2) {$e_{11}$};
    \node [style=rect] (6) at (-4.5, 0) {$e_1$};
    \node [style=rect] (7) at (-1.5, -1) {$e_4$};
    \node [style=rect] (8) at (-2, -2) {$e_{12}$};
    \node [style=rect] (9) at (-1, -2) {$e_{13}$};
    \node [style=rect] (10) at (0.5, -1) {$e_7$};
    \node [style=rect] (11) at (0, -2) {$e_{14}$};
    \node [style=rect] (12) at (1, -2) {$e_{15}$};
    \node [style=rect] (13) at (-0.5, 0) {$e_3$};
    \node [style=rect] (15) at (-2.5, 1) {$e_6$};

    \node (18) at (-6, -2.5) {$3$};
    \node (19) at (-5, -2.5) {$7$};
    \node (21) at (-2, -2.5) {$0$};
    \node (22) at (-4, -2.5) {$1$};
    \node (24) at (0, -2.5) {$2$};
    \node (26) at (-1, -2.5) {$4$};
    \node (27) at (-3, -2.5) {$5$};
    \node (28) at (1, -2.5) {$6$};
    \node (29) at (-5.5, -1.5) {$3$};
    \node (30) at (-1.5, -1.5) {$0$};
    \node (31) at (-3.5, -1.5) {$1$};
    \node (32) at (0.5, -1.5) {$2$};

    \node (33) at (-4.5, -0.5) {$1$};
    \node (34) at (-0.5, -0.5) {$0$};
    \node (35) at (-2.5, 0.5) {$0$};

    \draw [style=arr] (0) to (1);
    \draw [style=edge] (0) to (2);
    \draw [style=arr] (3) to (4);
    \draw [style=edge] (3) to (5);
    \draw [style=edge] (6) to (0);
    \draw [style=arr] (6) to (3);
    \draw [style=arr] (7) to (8);
    \draw [style=edge] (7) to (9);
    \draw [style=arr] (10) to (11);
    \draw [style=edge] (10) to (12);
    \draw [style=arr] (13) to (7);
    \draw [style=edge] (13) to (10);
    \draw [style=edge] (15) to (6);
    \draw [style=arr] (15) to (13);
\end{tikzpicture}
    \end{minipage}
    \caption{Complete binary tree with rotor initial state of pointers (left)
    and after \DetLong serves a request to element $e_6$ (right). 
    Each node is represented by a rectangle and its label denotes the element
    stored at this node. Arrows represent states of rotor pointers. The number
    below a~node denotes its flip rank: nodes with flip rank $0$ constitute the global
    path. \\
    Serving element $e_6$ induces the following changes: Elements $e_1$ and $e_2$ are moved
    one level down along the global path, $e_4$ is moved to the initial position of $e_6$, and 
    $e_6$ is moved to the root. The states of rotor pointers of the two topmost nodes on the global 
    path are flipped and the nodes' flip ranks are updated accordingly.}
    \label{fig:pd_example}
\end{figure*}
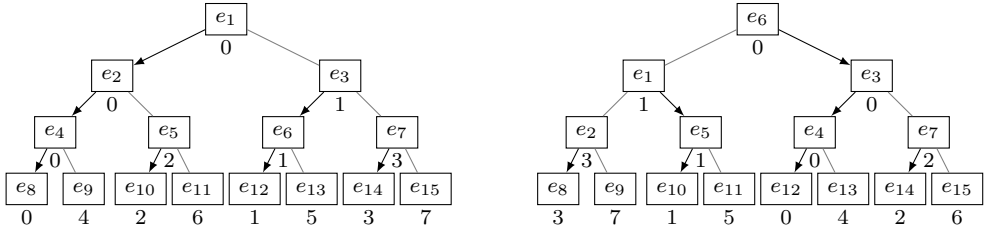

\noindent \textbf{Deterministic algorithm.}
Fix any complete binary tree $T$ with rotor pointers.
Upon a~request to an $\el^*$-level element $\re$, \DetLong (\DET)
fixes node $v = P^T_{\el^*}$ (possibly $v = \node(\re)$) 
and rearranges the elements by executing the augmented push-down operation $\PD{\node(\re)}{v}$. 
Then, it updates nodes' pointers executing $\flip^T(\el^*)$. 
An example tree reorganization performed by \DET is given in \autoref{fig:pd_example}.

\noindent \textbf{Access cost.}
Note that both algorithms ($\RAND$ and $\DET$), upon request to an element~$\re$ 
at level $\el^*$,
execute operation $\PD{\node(\re)}{v}$ for a node $v$ from level $\el^*$. 
Thus, their total cost can be bounded in the same way, 
by adding the access cost $\el^*+1$
to the swap cost of the augmented push-down operation.
The latter operation can be implemented efficiently. 

\begin{lemma}
\label{lem:access-cost}
It is possible to implement both considered algorithms (\RAND and \DET), so that 
they incur cost at most $4\cdot \el^*$ for a request to $\el^*$-level element $\re$.
\end{lemma}

\begin{proof}
If $\el = 0$, the observation holds trivially, and thus we assume that $\el \geq
1$. Either algorithm executes operation $\PD{u=\node(\re)}{v}$ for a node $v$
from level $\el^*$. Let $e = \elem(v)$. We first access element
$e$ (at cost $\el^*+1$). Then, we move $e$ to the root, swapping $\el^* - 1$ element
pairs on the path from $v$ to $r_T$. 
If $u = v$, then we are done. Otherwise, we move $e$ to node $u$, swapping
$\el^* - 1$ element pairs on the path from $r_T$ to $u$. At this point the element~$\re$ 
occupies the parent node of $u$. It remains to move it to the root,
swapping $\el^* - 2$ element pairs. In total, there are $3 \el^* - 4$ swaps. 
Adding the access cost of $\el^*+1$ yields the lemma.
\end{proof}

\iosif{For completeness we give the pseudocodes of the two remaining single-source tree network algorithms; \movehalf and \maxpush (\textsc{Strict-MRU}) \cite{avin2021dynamically}.}

 \LinesNumbered
 
\begin{algorithm}
\caption{\movehalf}
\label{alg:movehalf}
\textbf{access} $\sigma^t = e_i = el(u)$ along the tree branches\;

let $e_j=el(v)$ be the element with the highest rank at depth $\lfloor \ell(e_i)/2 \rfloor$\;

\textbf{swap} $e_i$ along tree branches to node $v$\; 

\textbf{swap} $e_j$ along tree branches to node $u$\;
\end{algorithm}

\begin{algorithm}
\caption{\maxpush (\textsc{Strict-MRU})}
\label{alg:maxpush}
\textbf{access} $\sigma^t = e$ at depth $k=\ell(e)$\;

move $e$ to the root\;

\For{depth $j=1, 2,\ldots, k-1$}{
move $e_j = \arg\max_{e\in E:\ell(e)=j} rank^{(t)}(e)$, i.e., the least recently used element in level $j$, 
to level $j+1$ to node $\node(e_{j+1})$\;}

move $e_k$ to $\node(e)$\;
\end{algorithm}


\section{Analysis of \DetLong}

We start with structural properties of rotor walks. In particular,
node pointers induce a~specific ordering of nodes on each level, 
which allows us to define their respective \emph{flip-ranks}.
Flip-ranks and levels play a crucial role in the amortized analysis of
\DetLong that we present in subsequent subsections.

\subsection{Flip-Ranks}

We say that 
a~node~$u$ is contained in the global path $P$ if $P_{\level(u)} = u$.

\begin{definition}[Flip-Ranks]
For any state of pointers in $T$ and a $\el$-level node $u$, $\frank^T(u) \geq
0$ is the smallest number of consecutive operations $\flip^T(\el)$ after which
$u$ is contained in~$P^T$. 
\end{definition}

It is easy to observe that when $\flip(\el)$ is executed $2^\el-1$ times, all
nodes of level $\el$ are at some point (i.e., before all flips or after one of
them) contained in $P^T$. That is, flip-ranks of $\el$-level nodes are distinct
numbers from the set $\{0, \dots, 2^\el-1\}$. 
An example of assigned flip-ranks is presented in \autoref{fig:pd_example}.
Furthermore, flip-ranks satisfy the
following recursive definition. (Recall that $T[u]$ is the tree rooted at $u$).

\begin{lemma}
\label{lem:recurrent_ranks}
Fix a tree $T$ and let a node $v$ be a descendant of a node $u$.
Then, $\frank^T(v) = \frank^T(u) + \frank^{T[u]}(v) \cdot 2^{\level(u)}$.
\end{lemma}
    
\begin{proof}
Observe that executing $\flip^T(\level(v))$
is equivalent to 
finding a node $w = P^T_{\level(u)}$ (on the same level as $u$) and then
\begin{itemize}
    \item executing $\flip^{T[w]}(\level(v)-\level(u))$ and
    \item executing $\flip^T(\level(u))$.
\end{itemize}
We will refer to operation $\flip^T(\level(v))$ simply as \emph{flip}. We now
compute $\frank^T(v)$, i.e., the number of flips after which $P^T$ contains $v$
for the first time. A necessary condition is that $P^T$ must contain its
ancestor $u$: this occurs for the first time after $\frank^T(u)$ flips, and more
generally after $\frank^T(u) + k \cdot 2^{\level(u)}$ flips, where $k \in
\mathbb{N}_{\geq 0}$. At each such time, pointers are toggled \emph{in the
subtree $T[u]$} (i.e., we execute operation
$\flip^{T[u]}(\level(v)-\level(u))$). It takes $\frank^{T[u]}(v)$ such
operations to make path $P^{T[u]}$ contain $v$, and thus the path $P^T$ contains~$v$ 
for the first time after $\frank^T(u) + \frank^{T[u]}(v) \cdot
2^{\level(u)}$ flips.
\end{proof}

\begin{lemma}\label{lem:flips}
Fix any state of pointers in $T$ and an $\el'$-level node $u$.
Fix level $\el$ and execute operation $\flip(\el)$.
\begin{itemize}
\item If $\el' \leq \el$, then the flip-rank of $u$
    becomes $2^{\el}-1$ if it was $0$ and decreases by $1$ otherwise.
\item If $\el' > \el$, then the flip-rank of $u$ can either increase 
    by $2^{\el}-1$ or decrease by $1$.
\end{itemize}
\end{lemma}

\begin{proof}
First assume $\el' \leq \el$. Note that the operation 
$\flip(\el)$ is equivalent to operation $\flip(\el')$ and toggling pointers 
of nodes $P_{\el'}, P_{\el'+1}, \ldots, P_{\el-1}$.
Thus, the first property follows immediately by the definitions of flip-ranks.

For the second part of the lemma, let $w$ be the $\el$-level ancestor of $u$.
As the pointers inside subtree $T[w]$ are unaffected by $\flip(\el)$,
$\frank^{T[w]}(v)$ remains unchanged.
Thus, by Lemma~\ref{lem:recurrent_ranks}, the change of $\frank^T(u)$ 
is exactly the same as the change of $\frank^T(w)$; by the previous argument
it can either grow by $2^\el-1$ or decrease by $1$.
\end{proof}


\noindent \textbf{Flip-ranks and Push-Down Operations.}
Finally, we can combine the effects of flip and push-down operations to determine
the way flip-ranks of elements change when \DetLong rearranges its tree.

\begin{observation}
\label{obs:elemRanksUpdate}
When \DetLong rearranges its tree upon seeing a request to an $\el^*$-level
element $\re$, then
\begin{enumerate}
    \item \label{item:path} 
        for all $\el < \el^*$, element $\elem(P_{\el}^T)$ is moved to level
        $\el+1$ and its flip-rank changes from $0$ to~$2^{\el+1}-1$,
    \item \label{item:path_last} 
        if $\elem(P_{\el^*}^T) \neq \re$, then its flip-rank changes 
        from $0$ to $\frank^T(\node(\re))-1$,
    \item \label{item:queried}
        element $\re$ is moved to the root and its flip-rank becomes $0$,
    \item \label{item:remaining}
    other elements remain on their levels, and their flip-ranks 
    may decrease at most by $1$.
\end{enumerate}
\end{observation}

\subsection{Credits and Analysis Framework}

From now on, we fix a single complete binary tree $T$. Thus, 
we drop superscript $T$ in notations $P$, $\flip$ and $\frank$ 
as it is clear from the context. While $\level(e)$ denotes the level 
of $e$ in the tree of $\DET$, we use $\levopt(e)$ to denote its level in the tree 
of \OPT.

We define \emph{level-weight} of $e$ as 
\begin{equation}
\label{eq:credlev}
    \credlev(e) = \begin{cases}
        \displaystyle \level(e) - 2 \cdot \levopt(e) - 1 & \text{if $\level(e) \geq 2 \cdot \levopt(e) + 2$}, \\
        0 & \text{otherwise},
    \end{cases}
\end{equation}
and \emph{(flip-)rank-weight} of $e$ as 
\begin{equation}
    \credpos(e) = \begin{cases}
        \displaystyle 1 - \frac{\frank(e)}{2^{\level(e)}} & \text{if $\level(e) \geq 2 \cdot \levopt(e) + 1$}, \\
        0 & \text{otherwise}.
    \end{cases}
\end{equation}
Finally, we fix $f = 4$ and let \emph{credit} of $e$ be
\[
    \cred(e) = f \cdot (\credlev(e) + \credpos(e)) .
\]
As at the beginning trees of \DET and \opt are identical, credits of all
elements are zero. Thus, our goal is to show that at any step the
\emph{amortized cost} of \DET, defined as its actual cost plus the
total change of elements' credits, is at most $O(1)$ times the cost of \opt. 
We do not strive at minimizing the constant hidden in the $O$-notation, but 
rather at the simplicity of the argument.

We split each round into two parts. In the first part, $\opt$ performs an
arbitrary number of swaps, each exchanging positions of two adjacent elements
and pays $1$ for each swap. In the second part, both \DET and \opt access a
queried element and \DET reorganizes its tree according to its definition.
Without loss of generality, we may assume that \opt does not reorganize its tree
in the second stage as it may postpone such changes to the first stage of the
next step.

In the following, we use $\DET$ and $\OPT$ to denote also their costs 
in the respective parts and we use $\Delta \cred(e)$, $\Delta \credlev(e)$,
and $\Delta \credpos(e)$ to denote the 
change in the credit and weights of element $e$ within considered part.


\subsubsection*{Part 1: OPT swaps}

\begin{lemma}
\label{lem:opt_swaps}
For any swap performed by \opt, it holds that $\sum_{e \in E} \Delta \cred(e)
\leq 3 \cdot f \cdot \opt$. 
\end{lemma}


\begin{proof}
Assume that \opt swaps a pair $(e_1, e_2)$, by moving $e_1$ one level down and
$e_2$ one level up. 
The weights associated with $e_1$ can
only decrease, and hence we only upper-bound $\Delta\cred(e_2)$. As $h(e_2)$
decreases by $1$, $\credlev(e_2)$ may grow at most by $2$ and $\credpos(e_2)$
may grow at most by $1$. 
Hence, $\Delta \cred(e_1) + \Delta \cred(e_2) \leq 3 \cdot f$.
This concludes the proof as \opt pays $1$ for the swap.
\end{proof}


\subsubsection*{Part 2: Requests are served}

We fix a requested element $\re$. 
and denote its level in the tree of \DET by $\el^* = \level(\re)$.
For $\el \leq \el^*$, we denote the $\el$-level element on the global path by
$e_\el$, i.e., $e_\el= \elem(P_\el)$. %
Recall that when \DET rearranges its tree, elements $\re, e_0, e_1, \dots,
e_{\el^*}$ change their respective nodes. We define three sets of elements:
$\{\re\}$, $P' = \{ e_0, e_1, \dots, e_{\el^*} \} \setminus \{\re\}$, 
and the set of remaining elements, denoted by $B$.
We first estimate the change in the elements' credits for sets $P'$ and $B$.

\begin{lemma}
\label{lem:e_k_moves}
It holds that $\sum_{e \in P'} \Delta \cred(e) \leq f$.
\end{lemma}

\begin{proof}
We first observe that if $e_{\el^*}$ is in the set $P'$, then it must be different
from $\re$. In such a case, it remains on its level, its flip-rank can only grow
(cf.~Case~\ref{item:path_last} of Observation~\ref{obs:elemRanksUpdate}), and thus
$\Delta \cred(e_{\el^*}) \leq 0$. 

In the following, we therefore estimate $\Delta \cred(e_{\el})$ for $\el <
\el^*$. The level of $e_\el$ increases by $1$ and its flip-rank changes from $0$ to
$2^{\el+1}-1$ (cf.~Case~\ref{item:path} of Observation~\ref{obs:elemRanksUpdate}).
We consider three cases.
\begin{itemize}

\item $\el \leq 2 \cdot h(e_\el) - 1$. Both $\credlev(e_\el)$ and $\credpos(e_\el)$
 remain zero, and thus $\Delta \cred(e_\el) = 0$.

\item $\el = 2 \cdot h(e_\el)$. Then, $\credlev(e_\el)$ remains zero, while 
     $\credpos(e_\el)$ increases from $0$ to $1 - (2^{\el+1}-1)/2^{\el+1} = 1/2^{\el+1}$.
     Thus, $\Delta \cred(e_\el) = f / 2^{\el+1}$.

\item $\el \geq 2 \cdot h(e_\el) + 1$. Then, $\credlev(e_\el)$ increases by $1$, while 
     $\credpos(e_\el)$ changes from $1-0/2^{\el} = 1$ to $1 - (2^{\el+1}-1)/2^{\el+1} = 2^{-\el-1}$.
     Thus, $\Delta \cred(e_\el) = f / 2^{\el+1}$. 
\end{itemize}
Summing up, we obtain 
$\sum_{e \in P'} \Delta \cred(e) \leq \sum_{\el=0}^{\el^*} \Delta \cred(e_\el) 
\leq f \cdot \sum_{\el=0}^{\el^*} 1 / 2^{\el+1} < f$.
\end{proof}

\begin{lemma}
\label{lem:single_level_shifts}
It holds that $\sum_{e \in B} \Delta \cred(e) \leq f$.
\end{lemma}

\begin{proof}
The node mapping of elements from $B$ remain intact, and thus their
level-weights are unaffected. However, their flip-ranks may change, although by
Observation~\ref{obs:elemRanksUpdate} (Case~\ref{item:remaining}) they may decrease at
most by $1$.

Fix any level $h \geq 0$ and let $B_h$ be the set of elements of $B$ on level
$h$ in the tree of \opt. For an element $e \in B_h$, if $\level(e) \leq 2h$, then the 
flip-rank-weight of $e$ remains $0$. If, however, $\level(e) \geq 2h + 1$,
then the flip-rank of $e$ decreases at most by $1$, and thus its flip-rank-weight
increases at most by $2^{-\level(e)}$. In total,
$ \sum_{e\in B_h} \Delta\credpos(e) 
 = \sum_{e\in B_h: \level(e) \geq 2h+1} 2^{-\level(e)}
 \leq \sum_{e\in B_h} 2^{-2h-1}
 \leq 2^{-h-1}.
$
The last inequality follows as $|B_h| \leq 2^h$.
Summing the above bound over all levels, we obtain 
$
 \sum_{e\in B} \Delta\credpos(e) 
 = \sum_{h=0}^{L_T} \sum_{e\in B_h} \Delta\credpos(e) 
 \leq \sum_{h=0}^{L_T} 2^{-h-1}
 < \sum_{h=0}^{\infty} 2^{-h-1}
 = 1.
$ 
Therefore, $\sum_{e \in B} \Delta\cred(e)
= f \cdot \sum_{e \in B} \Delta\credpos(e) \leq f$.
\end{proof}

\subsubsection*{Main Result}

\begin{theorem}\label{thm:12}
\DetLong is 12-competitive.
\end{theorem}

\begin{proof}
It is sufficient to show that within either part of a single round, $\DET +
\sum_{e \in E} \Delta \cred(e) \leq 12 \cdot \opt$. The theorem follows then by
summing this relation over all rounds, and observing that credits are zero
initially.

In the first part, when \OPT performs its swaps, the relation holds by
Lemma~\ref{lem:opt_swaps} as in this case $\DET + \sum_{e \in E} \Delta \cred(e)
\leq 0 + 3 \cdot f \cdot \OPT = 12 \cdot \OPT$.  

In the rest of the proof, we focus on the second part of the round.
By Lemma~\ref{lem:e_k_moves} 
and Lemma~\ref{lem:single_level_shifts}, 
the amortized cost of \DET in this part can be upper-bounded by 
\begin{align}
\nonumber
    \DET + \sum_{e \in E} \Delta\cred(e) 
    & \leq \DET + \Delta\cred(e^*) + \sum_{e \in P'} \Delta\cred(e) + \sum_{e \in B} \Delta\cred(e) \\
    & \leq \DET + \Delta\cred(e^*) + 2\cdot f.
    \label{eq:det_competitive}
\end{align}
It remains to bound $\DET + \Delta \cred(\re)$. 
To this end, let $h^* = \levopt(\re)$ be the level of $\re$ in the tree of $\OPT$.
By Lemma~\ref{lem:access-cost},
the cost of \DET is at most $4 \cdot \el^*$.
We consider two cases.
\begin{itemize}

\item $\el^* \leq 2 \cdot h^* + 1$. Then, the initial and the final credit of $\re$ is zero, 
and thus 
$\DET + \Delta\cred(\re) = 4 \cdot d^* \leq 8 \cdot h^* + 4$.

\item $\el^* \geq 2 \cdot h^* + 2$. The initial credit of $\re$ is
$\cred(\re) \geq f \cdot \credlev(\re) = (\el^* - 2 \cdot h^* - 1) \cdot f$ and
the final credit of $\re$ is zero. Thus, using $f = 4$, we obtain
$\DET + \Delta\cred(\re) \leq 4 \cdot \el^* - f \cdot \el^* + 2 \cdot  f \cdot h^* + f 
= 8 \cdot h^* + 4$.
\end{itemize}

Plugging the relation $\DET + \Delta\cred(\re) \leq 8 \cdot h^* + 4$ to
\eqref{eq:det_competitive}, using that the cost of \opt is $h^*+1$ and $f = 4$,
we obtain $\DET + \sum_{e \in E} \Delta\cred(e) \leq 8 \cdot h^* + 4 + 2 \cdot f
\leq 12 \cdot (h^* + 1) = 12 \cdot \OPT$.
\end{proof}


\subsection{On the Lack of Working Set Property}\label{sec:lack}

The next Lemma shows formally that the \DetLong does not maintain the working set property. This was first observed informally in \cite{techreport}.

\begin{lemma}
\label{lem:rotorWS}
\DetLong does not guarantee the working set property. The access cost of an element can be linear in its working set size.
\end{lemma}

\begin{proof}
We construct a sequence $\sigma$ of requests for which at some times the access cost in \DetLong will be linear in the working set size of the requested element. 
Consider a complete binary tree $T$ of size $2^x-1$ and $x$ levels, $0 \le \ell \le x-1$. Initially all pointers points to the left. Let $S$ be the set of nodes consisting of the root and the two left most nodes in each level. Clearly $|S|=2x-1$. 
We construct $\sigma$ by requesting only elements hosted by nodes in $S$. 
At each time the next request is to $\elem(v)$ where $v$ both in $S$ and $P^T$ and $\level(v)$ is the maximum possible. Formally $\el^* = \arg\max_{\el} P_\el^T \in S$ and $v=P_{\el^*}^T$.

Note that all elements that move during a request are moving between nodes in $S$. Therefore, the working set size is at most $2x-1$ for each request. 
The first request in the sequence is to element $e=\elem(P_{x-1}^T)$ and $e$ is moved to the root. It is not hard to verify
that for each level $\ell=\level(e) < x-1$ after a finite time $e$ will be pushed to level $\ell+1$.
Therefore after a finite time $e$ will reach level $x-1$ and will be requested again.
At that point the access cost will be $x$ while the working set property require a cost 
of $O(\log(2x-1))$.
%
\end{proof}






\section{Improved Analysis of \RandLong}
\label{sec:randompush}

In this section, we present a greatly simplified analysis of 
the algorithm \RandLong (\RAND)~\cite{avin2021dynamically},
showing that it is $O(1)$-competitive. 

We reuse the notation for the argument for \DetLong. 
We define level-weight of element $e$ as for \DetLong (see~\eqref{eq:credlev}). 
This time, however, we do not use flip-rank-weights, but we define 
the credit of element $e$ as 
$c(e) = f_R \cdot \credlev(e)$,
where $f_R = 8$.
We split the analysis of a single round, where an element $\re$ is requested,
again into two parts, where the swaps of \OPT are performed only in the former 
part. 

The proof for the following bound is analogous to Lemma~\ref{lem:opt_swaps},
but we get a slightly better bound as we need to analyze the growth of level-weights only. 

\begin{lemma}
\label{lem:randSwap}
For any swap performed by \opt, it holds that $\sum_{e \in E} \Delta \cred(e)
\leq 2 \cdot f_R \cdot \opt$. 
\end{lemma}

Throughout the rest of the proof, 
we fix a single requested element $\re$ and denote its level by $\el^*$. Our goal is to
prove that in the considered round
\begin{equation}
    \label{eq:randServiceCost}
    \textstyle \E[\RAND] + \E [\sum_{e\in E} \Delta \cred(e)] \leq 16 \cdot \opt.
\end{equation}
where the expected value is taken over random choices of an algorithm 
from the beginning of an input till the current round (inclusively).

Let $E' = E \setminus \{\re\}$. We first focus on the expected change of credits
in $E'$. 

\begin{lemma}
\label{lem:randPush}
It holds that $\E[\sum_{e \in E'} \Delta \cred(e)] \leq (\el^*/2 + 1) \cdot f_R$ 
\end{lemma}

\begin{proof}
We show a stronger property, namely that the lemma holds even if we fixed the mapping 
of elements to nodes (functions $\elem$ and $\node$) before the round. That is, we show  an upper bound the expected growth of credits, 
conditioned on an arbitrary fixed mapping and using only the randomness stemming 
from the choice of a random path chosen in the considered round. 

In particular, we assume that the level $\el^*$ of requested element $\re$ is fixed.
Recall that to serve $\re$, \RAND performs an augmented push-down operation 
along a random path of nodes $v_0, v_1, \dots, v_{\el^*}$, where $\level(v_i) = i$.
Let $E'_d$ be the set of elements of $E'$ on level $\el$. We upper-bound the value 
of $\E[\sum_{e \in E'_\el} \Delta \cred(e)]$. This value is clearly $0$ for 
$\el \geq \el^*$, as elements from such sets $E'_d$ do not change their levels. 
(Element $\elem(v_{\el^*})$ might be moved to another node, but remains on level $\el^*$.)
Furthermore, as at most one element from level $d$ increases its level (and its level-weight
can thus grow by at most $1$), $\E[\sum_{e \in E'_0} \Delta \cred(e)]$ 
and $\E[\sum_{e \in E'_1} \Delta \cred(e)]$ can be trivially upper-bounded by $f_R$ each.
Thus, we fix any level $\el \in \{2, \dots, \el^*-1\}$ and we look where the elements of $E'_d$
are stored in the tree of \OPT: let $A_\el \subseteq E'_\el$ be those elements whose level 
in the tree of \OPT is at most $\el-2$. To bound $\E[\sum_{e \in E'_\el} \Delta \cred(e)]$, we 
consider two cases.
\begin{itemize}
\item $e \in E'_\el \setminus A_\el$. 
Even if the level of $e$ increases to $d+1$ because of the augmented 
push-down operation, using $d \geq 2$, we have 
$\level(e) \leq d + 1 < 2 \cdot (d-1) + 2 \leq 2 \cdot \levopt(e) + 2$. Thus, 
by the definition of level-weight (see~\eqref{eq:credlev}), the credit of $e$ remains $0$ 
and $\Delta \cred(e) = 0$.

\item $e \in A_\el$. 
The growth of level-weight of $e$ is upper-bounded by $1$ and thus the increase of its 
credit upper-bounded by $f_R$. The increase however happens only if $\node(v_\el) = e$. 
As $v_\el$ is chosen randomly within level $\el$, this probability is equal to $1/2^\el$, 
and therefore $\E[\Delta \cred(e)] \leq f_R \cdot 2^\el$. 
\end{itemize}

Summing up, by the linearity of expectation, 
$$\E[\sum_{e \in E'} \Delta \cred(e)] 
    =  \E[\sum_{e \in E'_0} \Delta \cred(e)] 
    + \E[\sum_{e \in E'_1} \Delta \cred(e)] 
    + \sum_{\el = 2}^{\el^*-1} \E[\sum_{e \in E'_\el} \Delta \cred(e)]$$
Thus, 
\begin{align*}
\textstyle \E[\sum_{e \in E'} \Delta \cred(e)] 
    & \textstyle \leq   2 \cdot f_R
    + \sum_{\el = 2}^{\el^*-1} \sum_{e \in A_\el} \E\left[ \Delta \cred(e)\right] \\
    & \textstyle = 2 \cdot f_R 
    + \sum_{\el = 2}^{\el^*-1} |A_\el| \cdot f_R \cdot 2^{-\el}.
\end{align*}
Using $|A_d| = 2^{d-1} - 1$, we get
$\E[\sum_{e \in E'} \Delta \cred(e)] 
\leq 2 \cdot f_R + (\el^* - 2) \cdot f_R / 2 = (\el^* / 2 + 1) \cdot f_R$.
\end{proof}

The result now follows by combining the above lemmas essentially in the same way as we 
did in Theorem~\ref{thm:12} for \DET: a simple argument shows that 
the decrease of $\cred(e^*)$ is able to compensate for $\E[\RAND]$ and the increase
of remaining credits. 

\begin{theorem}\label{thm:16}
Algorithm \RandLong is $16$-competitive.
\end{theorem}

\begin{proof}
Similarly to the proof for \DET, 
it is sufficient to show that within either part of a~single round, 
$\E[\RAND] + \E[\sum_{e \in E} \Delta \cred(e)] \leq 16 \cdot \opt$. 
In the first part, when \OPT performs its swaps, the relation follows immediately by
\autoref{lem:randSwap}, thus below we focus on the second part of the round only.

Let $h^* = \levopt(\re)$ be the level of $\re$ in the tree of $\OPT$. 
By \autoref{lem:access-cost}, the cost of \RAND is at most $4\cdot \el^*$,
and by \autoref{lem:randPush}, the expected increase of credits from $E'$ is at most
$(\el^* / 2 + 1) \cdot f_R$. 
It remains to estimate the value of $\E[\Delta \cred(\re)]$. 
As $\re$ is moved to the root, its final credit is $0$. 
However, its initial credit depends on the value of $h^*$. 
\begin{itemize}
\item If $\el^* \leq 2\cdot h^* + 1$, then the initial credit of $e$ is $0$, and thus 
\begin{align*}
    \E[\RAND] + \E\left [\sum_{e\in E} \Delta \cred(e)\right ] =&
    4\cdot \el^* + (\el^* / 2 + 1) \cdot f_R\\=&  8 \cdot \el^* + 8 \leq 16 \cdot h^* + 16.
\end{align*}

\item If $\el^* \geq 2\cdot h^* + 2$, then the initial credit of $\re$ is
$(\el^* - 2\cdot h^* - 1) \cdot f_R$, and thus
\begin{align*}
    \E[\RAND] + \E\left [\sum_{e\in E} \Delta \cred(e)\right ] \leq\\
      4\cdot \el^* + (\el^* / 2 + 1) \cdot f_R - f_R \cdot \el^* + 2 \cdot h^* \cdot f_R + f_R \leq \\
      8 \cdot \el^* + 8 - 8 \cdot \el^* + 16 \cdot h^* + 8 = 16 \cdot h^* + 16.
\end{align*}
\end{itemize}
Thus, using that the cost of \OPT is $h^* + 1$, in 
either case, we obtained 
$\E[\RAND] + \E\left [\sum_{e\in E} \Delta \cred(e)\right ] \leq 16 \cdot h^* + 16
= 16 \cdot \OPT$, which concludes the proof.
\end{proof}


\section{Empirical Evaluation}
\label{sec:eval}

Although we have proven dynamic optimality for \RandLong and \DetLong, the question of which of the existing single-source tree network algorithms performs best in practice remains. 
In this section we turn to answer this question by empirically studying six algorithms: all the known
single-source tree network algorithms, i.e., \DetLong, \RandLong, \movehalf, and \maxpush\ (cf. Section \ref{sec:algorithms}), as well as the static offline balanced tree\footnote{\iosif{A static tree where elements are placed in decreasing frequency in a BFS order. \statopt\ performs no adjustments.}} (\statopt), and the demand-oblivious initial tree that performs no adjustments (\oblstat).

We compare all algorithms with synthetic and real access sequences with varying degrees of temporal and spatial locality.
Specifically we address the following five questions:
\begin{enumerate}[(Q1)]
\item \iosif{How does the benefit of self-adjustment depend on the network size?}\label{q:nwsize}

\item Which algorithm performs best with increasing temporal locality? \label{q:temporal}

\item Which algorithm performs best with increasing spatial locality?\label{q:spatial}

\item How does \DetLong compare to \RandLong in combined settings of temporal and spatial locality and how does it compare to \oblstat?\label{q:rotor}

\item \iosif{Do experiments with real data reflect the insights gained from those with synthetic data (Q\ref{q:nwsize}--Q\ref{q:rotor})?}\label{q:real}
\end{enumerate}

We elaborate on our empirical evaluation by presenting \iosif{our assumptions on locality} and methodology in Section \ref{sec:methodology}, our results together with their implications in  Section \ref{sec:results}, \iosif{and the main takeaways in Section \ref{sec:evaldiscussion}}.

\subsection{Methodology}
\label{sec:methodology}

We implemented all algorithms and the experimental setup in Python 3.9.
We tested all algorithms with synthetic and real data of varying locality.
Our source code and test data are publicly available \cite{repo}.
The initial trees were always constructed by placing the nodes uniformly at random.
\alenex{In Q\ref{q:temporal}--Q\ref{q:rotor}, we tested the scenario of 65,535 nodes  (complete binary tree of depth 15) and $10^6$ requests thoroughly, but we also experimented with different tree sizes. Our experiments showed that focusing on one scenario is representative of the algorithms behaviour.}
We repeated each experiment with synthetic data ten times and plotted the average values of the ten experiments for each case (per plot details follow).

\noindent \textbf{Temporal Locality}.
\iosif{Following \cite{10.1145/3379486}, we relate the degree of temporal locality of a sequence with the probability of repeating request $\sigma_i$, i.e., $p = \Pr[\sigma_{i+1} = \sigma_i]$.
Given $p$, we start by generating a sequence $\sigma$ of requests drawn uniformly at random. Then we post-process the sequence by the following rule: for $i=2,\ldots, 10^6$ with probability $p$, we set $\sigma_i = \sigma_{i-1}$ and otherwise $\sigma_i$ stays intact.}

\noindent \textbf{Spatial Locality}. \iosif{We used the Zipf distribution \cite{siegrist2017probability} (discrete, power law distribution)
to generate sequences of increasing spatial locality and decreasing empirical entropy.
In our context, a sequence with high spatial locality draws most requests from a small subset of nodes (the subset decreases as the  skewness increases), but requests for any node are allowed as well.
The probability mass function is $f(k, a) = 1/(k^a\sum_{i=1}^N i^{-a})$, for an element with weight 
$k$ and parameter $a$, where $N$ is the number of nodes and $a$ defines the skewness.
We set the weight of the $i^\text{th}$ element to $i^{-a}$ and normalized all weights. 
These sequences differ from the ones with controlled temporal locality in that we don't have any guarantees on the probability of repeating the previous request.}

\iosif{For Q\ref{q:nwsize}, we run experiments for trees with sizes 255, 1023, 4095, 16383, and 65535 nodes (i.e. tree depths 7, 9, 11, 13, 15) and $10^6$ requests. 
We computed the difference of the average total cost of each of the four self-adjusting algorithms minus the total cost of \oblstat,  
in high temporal ($p=0.9$) and spatial ($a=2.2$) locality scenarios. 
} 

For Q\ref{q:temporal}, we generated synthetic request sequences with increasing temporal locality. 
For each value of $p \in (0, 0.15, 0.3, 0.45, 0.6, 0.75, 0.9)$, the respective (average per ten samples for every case) empirical entropies\footnote{The empirical entropy of a sequence $\sigma$ is defined using the frequency $f(\sigma_i)$ of each element $\sigma_i$ in  $\sigma$: $\sum_{\sigma_i} f(\sigma_i)\log_2(1/f(\sigma_i))$ \cite{schmid2015splaynet}.} were $(15.95, 15.94, 15.91, 15.87, 15.81, 15.67, 15.16)$.
Thus, by increasing $p$ we indeed increase the degree of temporal locality of the sequence $\sigma$.

For Q\ref{q:spatial}, we defined a standard Zipf distribution over a fixed set of $N=65,535$ nodes and changed the distribution parameter to increase skewness.
For $a$ we used values from $(1.001, 1.3, 1.6, 1.9, 2.2)$, where the distribution skewness increases with $a$.
For each $a$ we drew sequences of length $10^6$ with respective empirical entropies $(11.07, 6.47, 3.88, 2.63, 1.92)$.

For Q\ref{q:rotor}, we focused on the performance of \DetLong, as it had the best performance in Q\ref{q:temporal} and Q\ref{q:spatial}, together with \RandLong. We first considered 65,535 nodes and $10^6$ requests that we constructed by combinations of temporal and spatial locality scenarios. 
We started with sequences drawn from Zipf distributions for $a\in \{1.001, 1.3, 1.6, 1.9, 2.2\}$ (as in Q\ref{q:spatial}), which we post-processed as in Q\ref{q:temporal}: we repeated the next element with probability $p\in \{0, 0.25, 0.5, 0.75, 0.9\}$.
For each sequence (defined by $a$ and $p$), we computed the average (total) cost difference between \DetLong and the oblivious static initial tree.
We repeated each experiment ten times and computed the averages.

We constructed a three-dimensional plot, where x-axis includes the values of $p$ (temporal locality), the y-axis includes the values of $a$ (spatial locality), and the z-axis shows the corresponding cost difference (\DetLong minus \oblstat).
We plotted the cost in a wireframe, where the data points form a grid (cf. Section \ref{sec:results} and Figure \ref{fig:Q4wire}).
Moreover, for ten sequences of $10^6$ requests drawn uniformly at random from the set of 65,535 nodes, we plotted a histogram of the cost differences of \DetLong and \RandLong, to show the extent to which they differ.


\iosif{For Q\ref{q:real} we used data from the Canterbury corpus \cite{arnold1997corpus} (as in \cite{albers2008list}). We used five books with the largest number of words. 
To increase the dataset sizes, we considered the string containing the sequence of words as they appear in each book, from which we extracted a sequence of requests by a sliding window of three letters, sliding by one character. That is, the first triple includes letters 1 to 3, the second 2 to 4, and so on, until the last three letters. The set of nodes (elements) for each sequence is derived by the set of unique triples appearing in each sequence.
Following this methodology for the five largest books of the corpus, we got (7,218; 6,962; 8,873; 6,225; 10,303) nodes and (3,128,781; 590,592; 261,829; 361,994; 1,627,137) requests, respectively.}

\iosif{To get an indication of the locality of these datasets we plotted them on a complexity map as it was defined in \cite{10.1145/3379486}. A complexity map shows the pairs of temporal and non-temporal complexity of each dataset. These quantities are computed using the size of compressed files, each containing a variant of the original sequence reflecting the two complexity dimensions. This method is different from the definitions of locality that we used in the synthetic data experiments and hence serves only as an indication.
}

\subsection{Results}
\label{sec:results}

We demonstrate and discuss our results for Q\ref{q:nwsize}--Q\ref{q:real}. 

\noindent \textbf{Q\ref{q:nwsize}: Network Size and Adjustment Benefit} \iosif{In figures \ref{fig:increasing-size-temp0.9} and \ref{fig:increasing-size-spatial2.2} we can see that as the tree size increases the benefit of reconfiguration increases as well.
This is expected as in larger trees, requests of non-frequent elements are more expensive and adjustment is more beneficial.
Therefore, in the following plots, the thresholds after which our adaptive algorithms perform better than \statopt, are not absolute, as they improve with network size.}

\newcommand{\diffplotsize}{0.3}
\begin{figure}[t]
    \centering
    \subfloat[\centering $p = 0.9$]{{\includegraphics[scale=\diffplotsize]{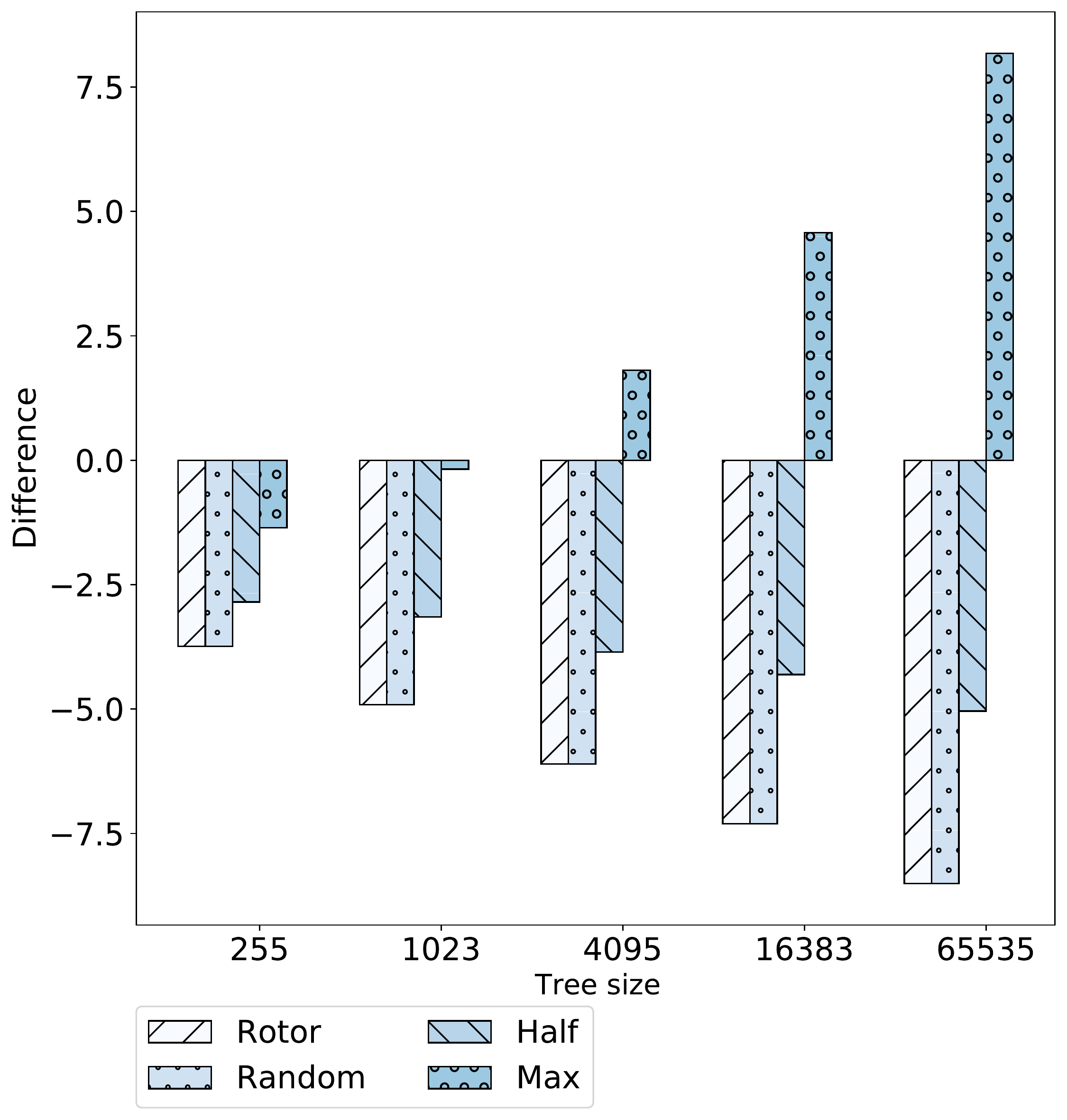} }\label{fig:increasing-size-temp0.9}}%
    ~~~~
    \subfloat[\centering $a = 2.2$]{{\includegraphics[scale=\diffplotsize]{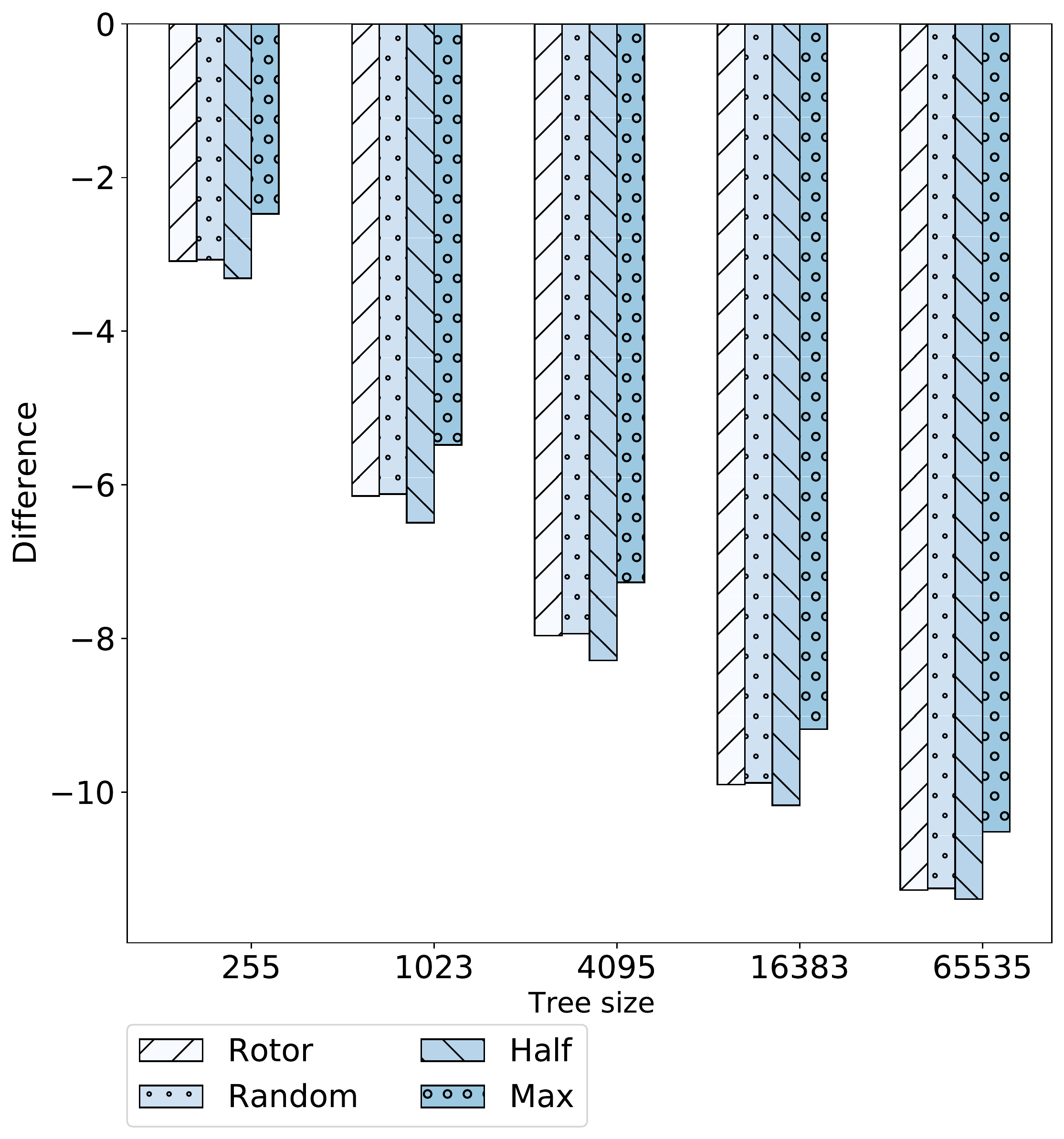} }\label{fig:increasing-size-spatial2.2}}    
    \caption{Q\ref{q:nwsize}: Total cost difference of the self-adjusting algorithms minus \oblstat, for high temporal ($p = 0.9$) and spatial ($a = 2.2$) locality.}%
    \label{fig:diff}%
\end{figure}



\begin{figure}[t]
    \centering
    \begin{minipage}{.49\textwidth}
        \centering
        \includegraphics[scale=0.33]{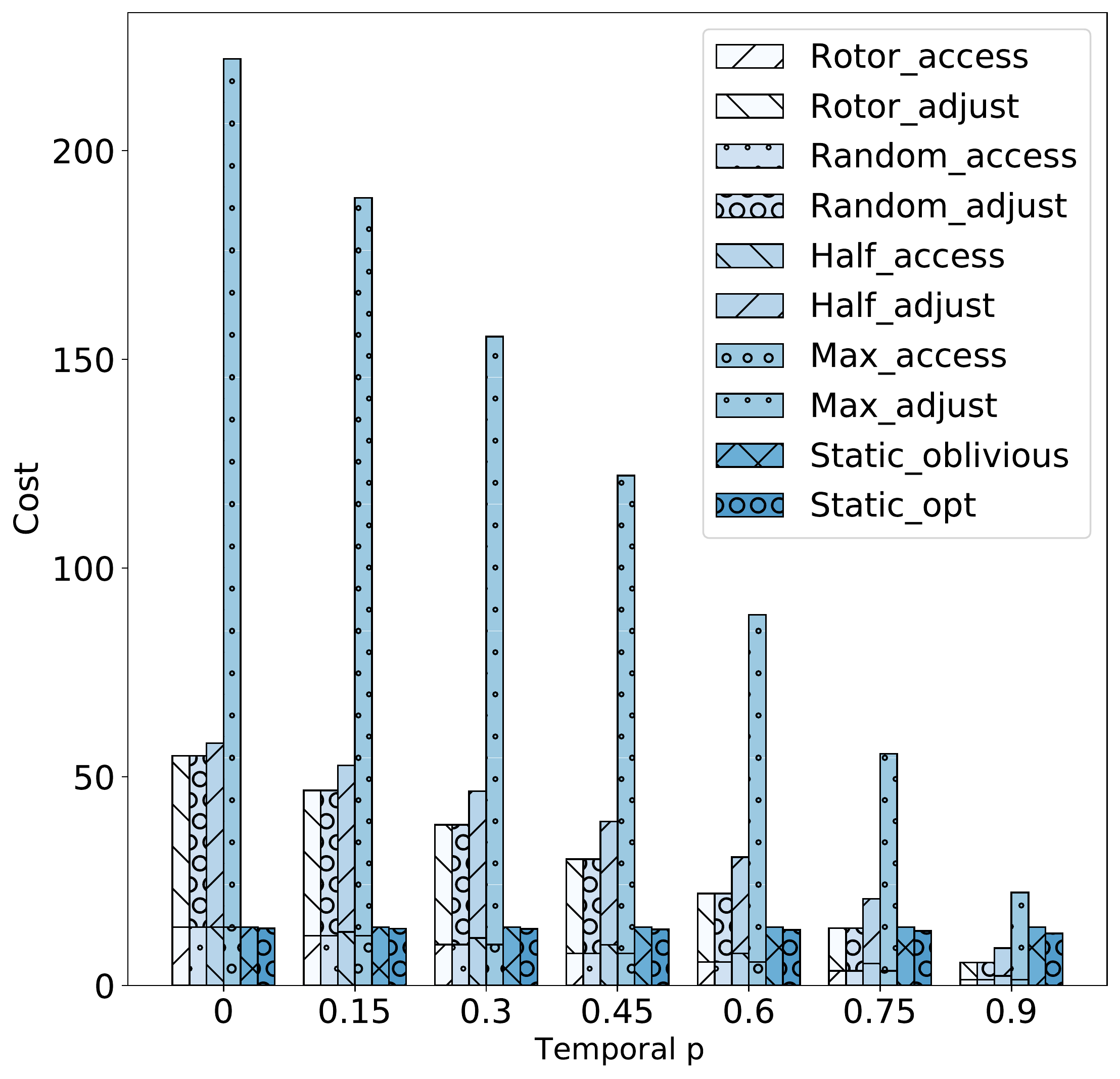}
        \caption{Q\ref{q:temporal}: results for temporal locality. The x-axis shows the  probability of repeating the last element and y-axis shows the costs.}
        \label{fig:Qtemporal}
    \end{minipage}%
    \hfill
    \begin{minipage}{0.49\textwidth}
        \centering
        \includegraphics[scale=0.33]{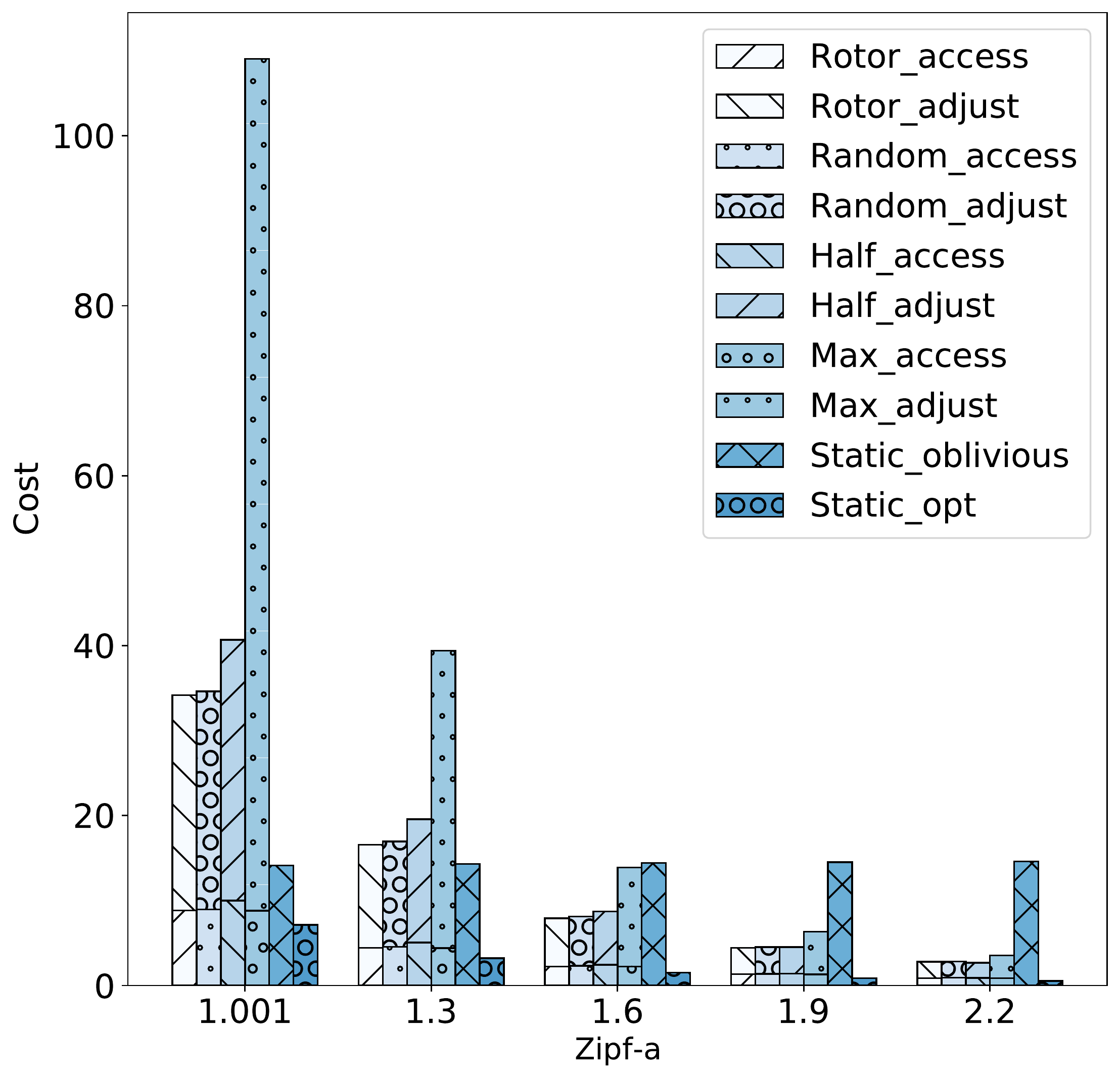}
        \caption{Q\ref{q:spatial}: results for spatial locality. The x-axis shows the Zipf distributions parameters and the y-axis the average cost. 
    }
        \label{fig:Qspatial}
    \end{minipage}
\end{figure}

\begin{figure}[h!]
    \centering
    \subfloat[\centering \DetLong minus \oblstat]{{\includegraphics[scale=0.7]{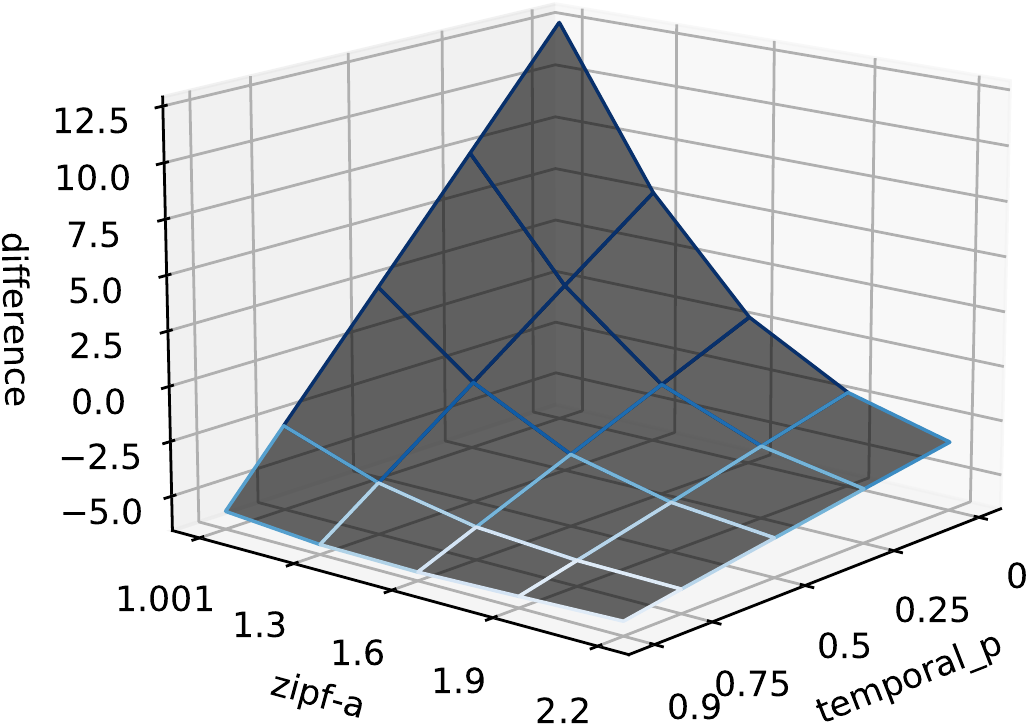} }\label{fig:Q4wire}}%
    ~~~~~
    \subfloat[\centering \DetLong minus \RandLong]{{\includegraphics[scale=0.45]{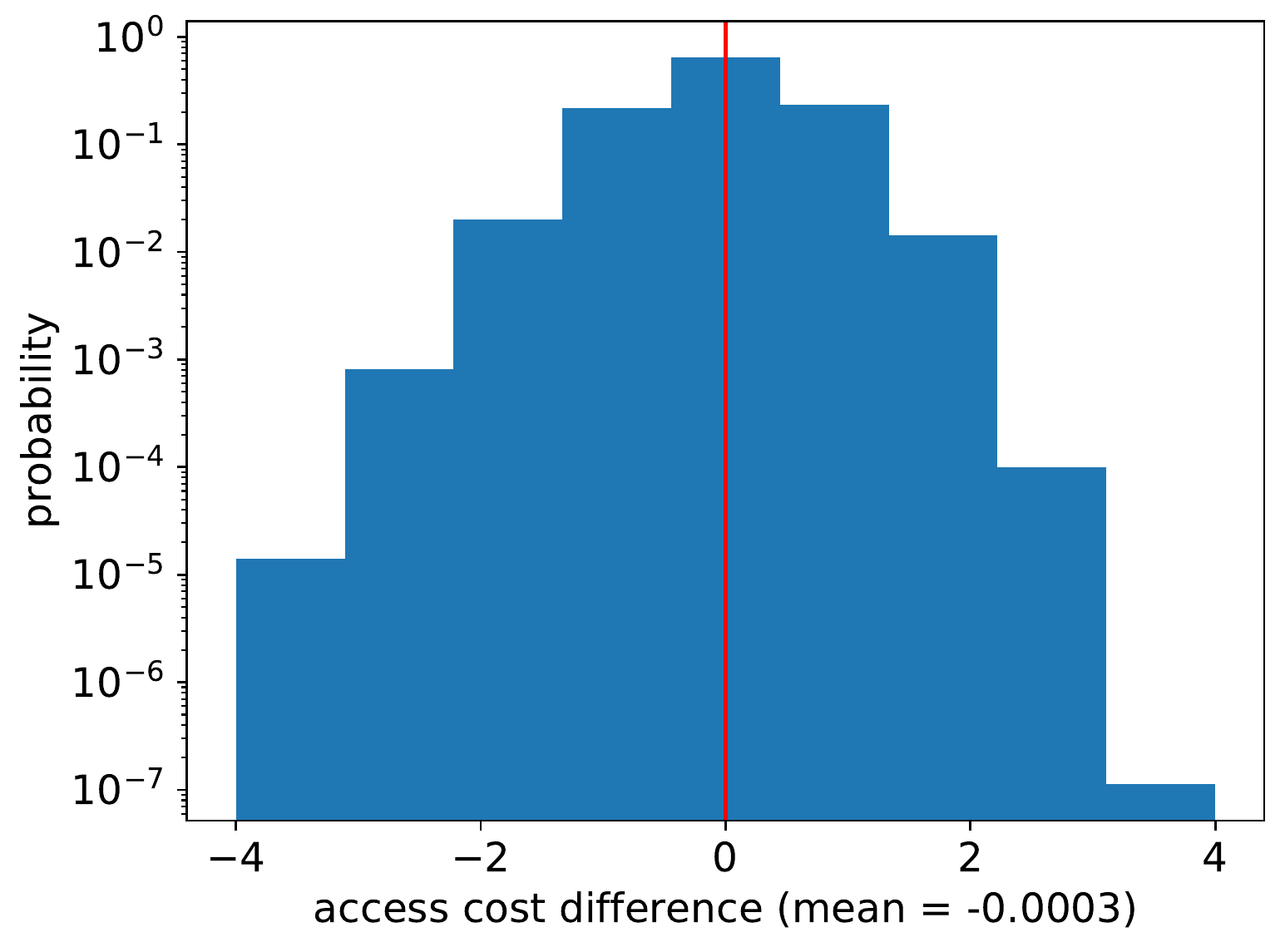} }\label{fig:Q4rotrand}}    \caption{Q\ref{q:rotor}: \DetLong performance. Figure \ref{fig:Q4wire} shows the difference of average cost between \DetLong and \oblstat in combined scenarios of temporal and spatial locality of Q\ref{q:temporal} and Q\ref{q:spatial} (negative values have lighter color). Figure \ref{fig:Q4rotrand} shows a histogram of the access cost difference distribution per request between \DetLong and \RandLong, taken over ten sequences of uniform data. The mean is -0.0003 (marked with a red vertical line). }%
    \label{fig:Qrotor}%
\end{figure}

\noindent \textbf{Q\ref{q:temporal}: Temporal Locality.}
In Figure \ref{fig:Qtemporal} we present our results for Q\ref{q:temporal}. 
We plotted the total cost for each algorithm.
We observe that \DetLong and \RandLong have the best performance and that all self-adjusting algorithms exploit temporal locality, as expected, but with varying efficiency. 
Interestingly, \DetLong and \RandLong outperform all other algorithms a bit after $p=0.75$, while \movehalf is only marginally more costly. On the other hand, the adjustment cost of \maxpush is quite high in all scenarios.

\noindent \textbf{Q\ref{q:spatial}: Spatial Locality.}
In Figure \ref{fig:Qspatial} we show our results for the spatial locality experiments. 
For the sequence of Zipf distributions with parameters $a \in (1.001, 1.3, 1.6, 1.9, 2.2)$, the respective average empirical entropies of the sequences that we sampled are $(11.07, 6.47, 3.88$, $2.63, 1.92)$. That is, as $a$ increases, the sequences are more skewed, and the entropy decreases.
Similarly to the temporal locality results, we observe that indeed all self-adjusting algorithms exploit spatial locality (\DetLong, \RandLong, and \maxpush have similar performance), and the reconfiguration cost pays off already from $a=1.6$ (when compared to \oblstat).
However, \statopt\ has the best performance in all scenarios.



\noindent \textbf{Q\ref{q:rotor}: \DetLong Performance.}
In Q\ref{q:rotor} we take a closer look on the performance of \DetLong, as in Q\ref{q:temporal} and Q\ref{q:spatial} it has the best performance, together with \RandLong.
In Figure \ref{fig:Q4wire} we plot the total cost difference between \DetLong and the oblivious static initial tree (\oblstat), in various scenarios of temporal and spatial locality.
As expected, their combination has a more dramatic effect in cost reduction.
Moreover, for ten sample sequences (each of length $10^6$) we observed (Figure \ref{fig:Q4rotrand}) that the difference between the cost of \DetLong and \RandLong is at most 4 (mean is $-0.0003$).
Thus, the variance in their performance difference is also rather small (in the previous sections we observed that the means are almost equal).




\begin{figure}[t]
    \centering
    \begin{minipage}{.49\textwidth}
        \centering
        \includegraphics[scale=0.5]{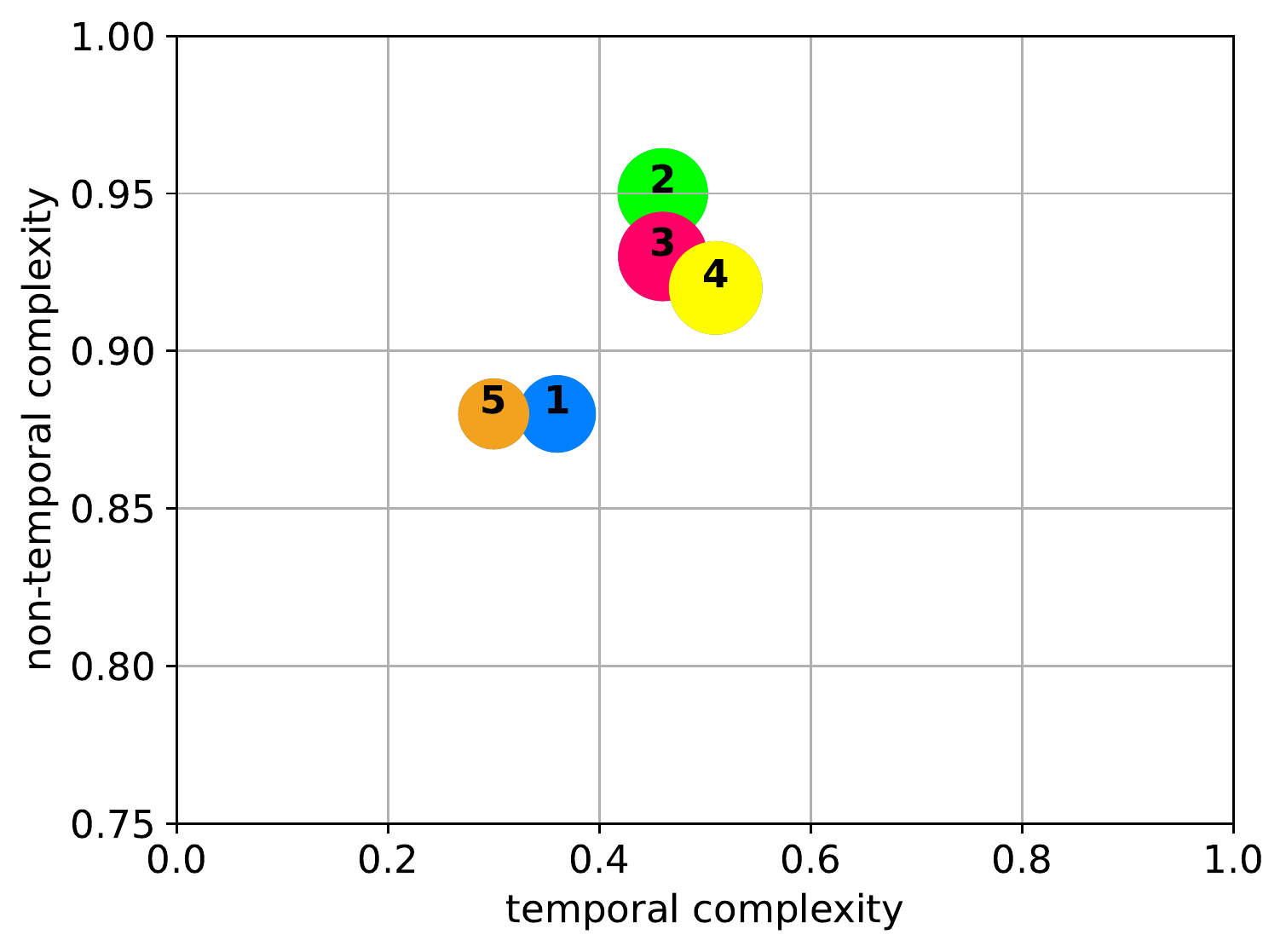}
        \caption{Q\ref{q:real}: Complexity map \cite{10.1145/3379486} of the five datasets extracted from the five largest books in the corpus data.}
        \label{fig:complmap}
    \end{minipage}%
    \hfill
    \begin{minipage}{0.49\textwidth}
        \centering
        \includegraphics[scale=0.33]{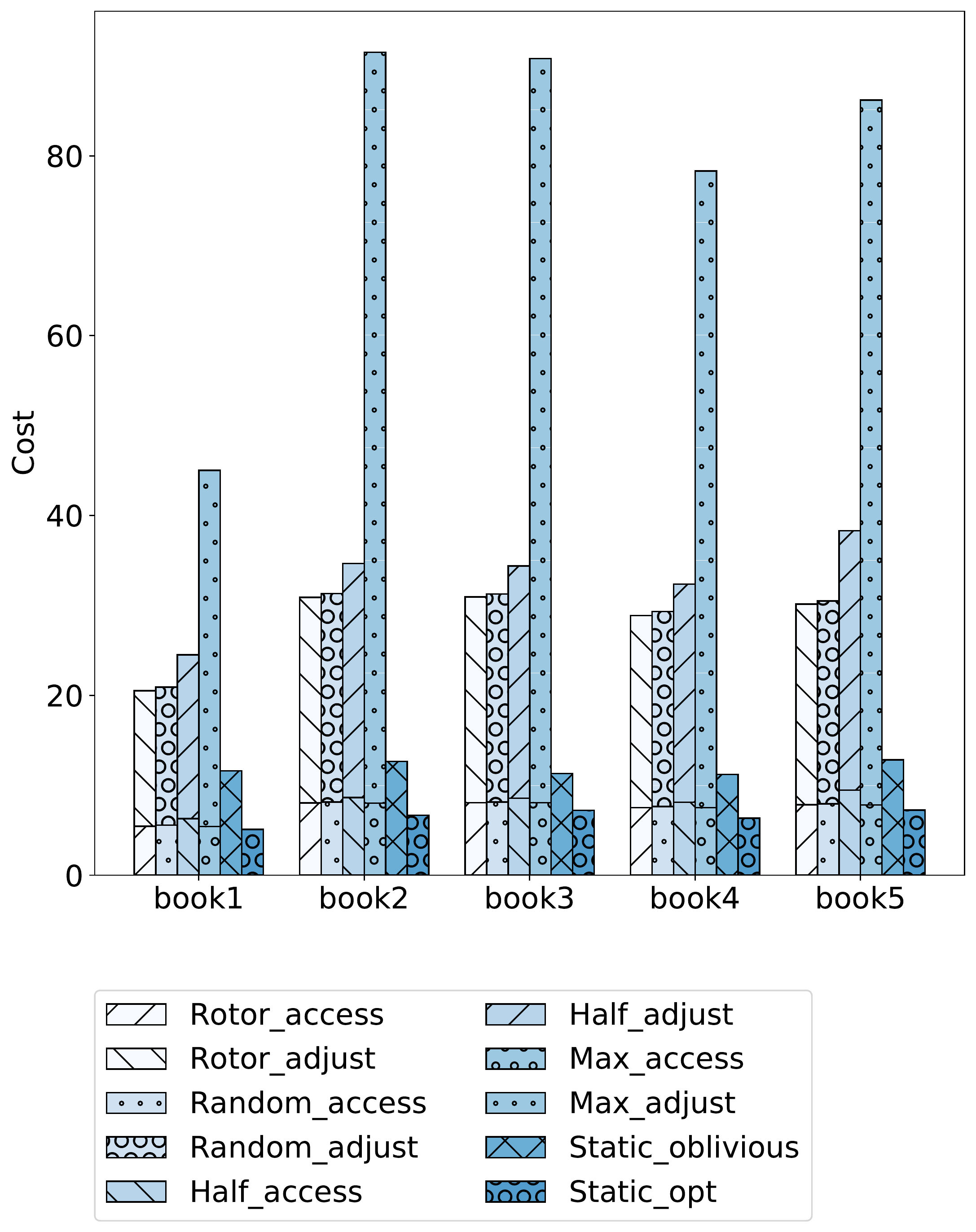}
        \caption{Q\ref{q:real}: Performance of the corpus data.}
        \label{fig:corpusdataplot}
    \end{minipage}
\end{figure}




\noindent \textbf{Q\ref{q:real}: Evaluation with corpus data.}
\iosif{
The complexity map computation \cite{10.1145/3379486} of the five datasets showed that their temporal complexity is in the interval $[0.3, 0.5]$ and their non-temporal complexity is in the interval $[0.8, 1]$ (Figure \ref{fig:complmap}). This plot indicates that the datasets have moderate to high locality.
In Figure \ref{fig:corpusdataplot} we plotted the performance of all six algorithms over these datasets.
As in the synthetic data, we observe that 
(i) \DetLong\ and \RandLong are the best self-adjusting algorithms with similar performance,
(ii) the access cost of \DetLong, \RandLong, and \movehalf\ is similar to the one of \statopt, and that
(iii) the selected dataset doesn't have high locality and hence the adjustment cost remains high.
}

\subsection{Discussion}
\label{sec:evaldiscussion}

\iosif{
We discuss the main takeaways of our evaluation. From the plots that address Q\ref{q:nwsize} we derived that in high locality scenarios, self-adjusting algorithms perform better as the network size increases, since the access cost for static algorithms increases as well (i.e. the tree size increases).
We then fixed the tree size to 65,535 nodes (depth 15) and observed that the cost of adjustment pays off in high locality sequences (temporal, spatial, or combined).
We observed that \DetLong\ and \RandLong\ have almost identical performance, both in synthetic and real data, despite their different properties. Recall that \RandLong has the working set property \cite{avin2021dynamically}, but \DetLong\ doesn't (cf. Section \ref{sec:lack}, Lemma \ref{lem:rotorWS}).
Specifically, even though the cost of \DetLong\ can be linear in the working set in theory, we did not observe this in any of the tested scenarios.  
Also, we found that the performance of all algorithms over corpus data follows the one observed with synthetic data.
}

\section{Future Work}
Our paper leaves open several interesting directions for future research.
On the theoretical front, it would be interesting to provide tight constant bounds on 
the competitive ratio of our algorithm and the problem in general.
On the applied front, it remains to engineer our algorithms further to
improve performance in practical applications, potentially also supporting concurrency.


\end{document}